\documentclass[letterpaper,11pt]{article}

\usepackage{amsfonts,amsmath,amsthm,amssymb,mathtools,dsfont}
\usepackage{geometry,color,appendix,graphicx,adjustbox,float,array,hyperref}
\usepackage{caption,subcaption}
\usepackage[boxed,figure]{algorithm2e} 

\providecommand{\U}[1]{\protect\rule{.1in}{.1in}}

\newcommand{\R}{\mathbb{R}}
\newcommand{\E}{\mathbb{E}\thinspace }

\newcommand{\PP}{\mathbb{P}\thinspace }

\newtheorem{theorem}{Theorem}

\newtheorem{corollary}[theorem]{Corollary}

\newtheorem{definition}[theorem]{Definition}
\newtheorem{example}[theorem]{Example}

\newtheorem{lemma}[theorem]{Lemma}

\newtheorem{proposition}[theorem]{Proposition}
\newtheorem{remark}[theorem]{Remark}

\numberwithin{equation}{section}
\numberwithin{theorem}{section}

\begin{document}
	
	\title{Optimal Insurance to Minimize the Probability of Ruin: \break Inverse Survival Function Formulation}
	
	\author{Bahman Angoshtari
		\thanks{
			1365 Memorial Drive, Ungar 523, Department of Mathematics, University of Miami,  Coral Gables, FL 33146, USA. Email: \texttt{b.angoshtari@math.miami.edu}.}
		\and
		Virginia R. Young
		\thanks{
			530 Church Street, Department of Mathematics, University of Michigan, Ann Arbor, MI 48109, USA. Email: \texttt{vryoung@umich.edu}.}}
	\maketitle
	
	\begin{abstract}
		We find the optimal indemnity to minimize the probability of ruin when premium is calculated according to the distortion premium principle with a proportional risk load, and admissible indemnities are such that both the indemnity and retention are non-decreasing functions of the underlying loss.  We reformulate the problem with the inverse survival function as the control variable and show that deductible insurance with maximum limit is optimal.  Our main contribution is in solving this problem via the inverse survival function.
		
		\medskip
		
		
		\medskip
		
		
	\end{abstract}

	\section{Introduction}
	
	The problem of computing optimal insurance has a long history, dating back to optimal risk exchanges studied by Borch \cite{B1960, B1960b} and optimal insurance design studied by Arrow \cite{A1963}.  In this introduction, we describe the work in papers related to the problem we solve, namely, minimizing the probability of ruin when premium is calculated according to the distortion premium principle.  Throughout this introduction, $X$ denotes the insurance loss, a non-negative random variable, $I(X)$ denotes the indemnity purchased by a buyer of insurance, and $R(X) = X - I(X)$ denotes the loss retained by the buyer.
	
	Gajek and Zagrodny \cite{GZ2004} assume that $0 \le I(x) \le x$ for all $x \ge 0$ and assume that the insurance premium is a non-decreasing function of $\E I(X)$.  Then, if the buyer of insurance wants full protection against ruin, then deductible insurance is the cheapest, which ties in nicely with Arrow's Theorem on the optimality of deductible insurance.  Gajek and Zagrodny \cite{GZ2004} also show that, if the buyer allows for ruin by limiting the amount spent on insurance, then {\it truncated} deductible insurance is optimal, optimal in the sense of minimizing the probability of ruin.  By {\it truncated deductible insurance}, we mean insurance of the form $I(x) = (x - d)_+$ when $x < m$ and $I(x) = 0$ when $x \ge m$, for some $m > d$.

	Kaluszka \cite{K2005} computes the insurance premium via an economic principle (using a price density), a generalized zero-utility principle, an Esscher principle, or a mean-variance principle.  Optimal insurance to minimize the probability of ruin is truncated deductible insurance for each of these premium principles, and he proves this result directly.  He assumes that the premium is fixed so does not find the optimal amount to spend on insurance.  Kaluszka \cite{K2005} also considers optimal insurance to minimize the probability of ruin if the insurance company invests in a risky asset during the period. He assumes that the insurance premium is computed via the expected-value principle with a positive safety loading; again, the amount spent on insurance is given {\it a priori}.  Truncated deductible insurance is optimal, and Kaluska also finds the optimal amount to invest in the risky asset.  Finally, Kaluszka \cite{K2005} considers a variety of other models: (1) He minimizes the variance of the insurance company's retained risk subject to a VaR constraint (equivalent to computing the premium via a percentile premium).  Optimal insurance is a modified deductible policy with a maximum limit (modified DIML policy).  (2) He reconsiders Arrow's problem with the addition of a VaR constraint.  Optimal insurance is a modified DIML policy.  (3) He maximizes the expected dividend payment when the insurance premium is computed according to the expected-value principle.  Optimal insurance is truncated full coverage.  (4) He maximizes a Kahneman-Tversky value function subject to a premium constraint based on the expected-value principle.  Optimal insurance is truncated deductible insurance.
	
	Cai et al.\ \cite{CTWZ2008} compute the insurance premium via the expected-value principle with a positive safety loading $\theta > 0$.  The class of indemnity functions are is the collection of non-decreasing, convex functions $I$ such that $0 \le I(x) \le x$.  They show that proportional deductible insurance minimizes both VaR and CTE, that is, $I^*(x) = c (x - d)_+$, which includes the proportion $c = 1$ and deductible $d = 0$.  Cheung \cite{C2010} revisits the problem considered by Cai et al.\ \cite{CTWZ2008} and uses a geometric approach to solve the minimization problem, which results in a simpler proof of the optimality of proportional deductible insurance.  Also, Chi and Tan \cite{ChiTan2011} simplify and extend the work of Cai et al.\ \cite{CTWZ2008}.  Specifically, they show that to minimize VaR of the ceding company when insurance is computed according to the expected-value principle, (1) proportional deductible insurance is optimal when $I(x)$ is restricted to be increasing and convex, (2) DIML is optimal when both $I(x)$ and $R(x) = x - I(x)$ are restricted to be increasing, (3) truncated deductible insurance is optimal when $R(x)$ is restricted to be increasing and left continuous.  By contrast, Chi and Tan \cite{ChiTan2011} show that to minimize CVaR, deductible insurance is optimal for all three cases.
	
	Cui et al.\ \cite{CuiYangWu2013} compute the insurance premium via a distorted premium principle with distortion function not necessarily equal to 1 at 1, and they restrict the class of indemnities such that both the indemnity and retention functions are non-negative and non-decreasing with the underlying loss.  Optimal insurance is a layered policy with the various layers depending on how the two distortion functions cross each other.  Assa \cite{Assa2015} extends the work of Cui et al.\ \cite{CuiYangWu2013} to consider the risk measures of both the buyer and seller of insurance.  He assumes that indemnity and retention functions are non-negative and non-decreasing, which implies that any such indemnity function $I(x)$ is absolutely continuous with respect to Lebesgue measure and, thus, is differentiable almost everywhere.  Hence, there exists a Lebesgue integrable function $h$ such that $I(x) = \int_0^x h(t) \, dt$.  This function $h$ is called the {\it marginal indemnification function} (MIF), and Assa \cite{Assa2015} uses the MIF to obtain the optimal indemnity.  Similarly, Zhuang et al.\ \cite{ZhuangWengTanAssa2016} use the MIF to extend the work of Cui et al.\ \cite{CuiYangWu2013} by further restricting $h$ so that $h_0 \le h(x) \le h_1$ for some constants $0 \le h_0 < h_1 \le 1$.

	Weng and Zhuang \cite{WengZhuang2017} minimize the probability of ruin when insurance is computed according to a (concave) distortion premium principle with safety loading $\theta > 0$ for a given amount spent on insurance $\pi = \Pi_g(I(X))$.  They assume that the retention $R(x)$ lies between $0$ and $x$ (as do all of these papers) and that $R$ is a non-decreasing function of the underlying loss.  Weng and Zhuang \cite{WengZhuang2017} inspired this paper because they reformulate their minimization problem in terms of the CDF of the retained loss.  We redo their work by assuming that $I$ is also an increasing function of the underlying loss, which is a natural requirement to prevent the moral hazard of underreporting losses.  Instead of using the CDF of $R(X)$, we use the inverse survival function, which we find more natural.  Weng and Zhuang \cite{WengZhuang2017} mistakenly write that $\pi = \Pi_g(X) - \Pi_g(R(X))$, which one cannot assert for their model because $I$ is not necessarily an increasing function of the underlying loss.  That is, as random variables, $I(X)$ and $R(X)$ are not necessarily comonotonic, which is a requirement in order to say that $\Pi_g(I(X)) + \Pi_g(R(X)) = \Pi_g(X)$.

	In this paper, we minimize the probability of ruin when premium is calculated according to the distortion premium principle with a proportional risk load, and we restrict both $I(x)$ and $R(x) = x - I(x)$ to be non-decreasing, thereby extending the work of Chi and Tan \cite{ChiTan2011} who minimize VaR.\footnote{Recall that Chi and Tan \cite{ChiTan2011} compute the insurance premium via the  expected-value premium principle, which is a special case of the distortion premium principles with a proportional risk load.} The remainder of the paper is organized as follows.  In Section \ref{sec:problem}, we present the model set up and state the minimization problem with the retention function as the control variable. Then, in Section \ref{sec:reform}, we reformulate the problem with the inverse survival function as the control variable and show that deductible insurance with maximum limit (DIML), or so-called {\it limited deductible insurance}, is optimal.  Finally, in Section \ref{sec:opt}, we obtain the optimal deductible and maximum limit.

	\section{Model set up and ruin-minimization problem}\label{sec:problem}
	
	Consider an individual with initial capital $w > 0$ who is subject to a loss represented by a non-negative random variable $X$ and who seeks to minimize her probability of ruin by buying insurance.
	
	For a random variable $Y$, let $S_Y$ denote the \emph{survival function} of $Y$, that is,
	\begin{equation}\label{eq:Sy}
		S_Y(x) = \PP(Y>x), \quad x \in \R,
	\end{equation}
	and define the corresponding \emph{inverse survival function}, or \emph{quantile function}, by
	\begin{equation}\label{eq:Sy_inv}
		S^{-1}_Y(p) = \inf\{x \ge 0: S_Y(x) \le p\}, \quad p \in [0,1].
	\end{equation}
	Note that $S^{-1}_Y$ is non-increasing and right-continuous, as is $S_Y$.
	
	Let $M$ denote the essential supremum of $X$, which might equal $\infty$; specifically,
	\begin{equation}\label{eq:M}
		M = \inf\{x\ge0: S_X(x)= 0\} \in \R^+\cup \{\infty\}.
	\end{equation}
	We assume throughout the paper that $S_X$ is continuously differentiable and strictly decreasing on $[0, M]$, and that $S_X(M):=\lim_{x\to M^-} S_X(x) = 0$. Note that $X$'s quantile function satisfies
	\begin{equation}\label{eq:SX_SXinv}
		S_X\big(S_X^{-1}(p)\big) = p,
		\quad\text{and}\quad
		S_X^{-1}\big(S_X(x)\big) = x,
	\end{equation}
	for all $p \in [0, S_X(0)]$ and $x \in [0, M]$.
	
	\begin{remark}
		$X$ can have an atom at $0$, in which case, $S_X(0) := \lim_{x \to 0^+} S_X(x) < 1$. Furthermore, $X$ can be unbounded from above, in which case $M = \infty$. With a slight abuse of notation, the interval $[a, M]$ and the function evaluation $f(M)$ are interpreted as $[a, \infty)$ and $\lim_{x \to \infty} f(x)$, respectively, when $M = \infty$.   \qed
	\end{remark}
	
	The individual buys insurance to mitigate her loss $X$.  We represent the insurance contract by an \emph{indemnity} $I$ such that the insurer is responsible for the amount $I(x)$ if the policyholder's loss $X = x$.  We also define the \emph{retention} $R$ by $R(x) = x - I(x)$; thus, when the loss $X = x$, $R(x)$ is the amount of loss for which the policyholder is responsible.  We impose the following conditions on the policyholder's retention (and indemnity).
	
	\begin{definition}\label{def:IndemnityFn}
		A function $R:[0,M] \to [0,M]$ is an {\rm admissible retention} if
		\begin{enumerate}
			\item[$(i)$] $0 \le R(x) \le x$ for all $x \in [0, M];$
			
			\item[$(ii)$] $R$ is non-decreasing on $[0, M];$ and,
			
			\item[$(iii)$] $I(x) = x - R(x)$, $x \in [0, M]$, is also non-decreasing.
		\end{enumerate}
		We denote the set of all admissible retentions by $\mathcal{A}$.  \qed
	\end{definition}

	Condition (i) states that neither the policyholder nor the insurer will be responsible for more than the total loss $X$.  Indeed, if $R(x) < 0$ for some $x \ge 0$, then the individual would be motivated to create or report loss $X = x$ because she would gain $I(x) - x > 0$ from the loss.  No insurance company would issue a policy that allowed a policyholder to gain from an insurable loss.  Also, if $R(x) > x$ for some $x \ge 0$, then the individual would not buy such a policy because, in addition to suffering a loss $X = x$, the individual would have to pay the insurance company $R(x) - x > 0$.
	
	Conditions (ii) and (iii) reduce the risk of moral hazard.  Indeed, if the retention were decreasing, then the policyholder would have an incentive to create more loss in order to reduce her retention.  Similarly, if the indemnity were decreasing, then the policyholder would have an incentive to report a smaller loss to the insurer.  Note that a retention $R$ belongs to $\mathcal{A}$ if and only if the indemnity $I$ also belongs to $\mathcal{A}$.
	
	For any $R \in \mathcal{A}$, we introduce the following notation. With a slight abuse of notation, we define $S_R:= S_{R(X)}$, that is, $S_R$ is the survival function of $R(X)$. We also denote the quantile function of $R(X)$ by $S^{-1}_R$.  Similarly, we write $S_I$ for $S_{I(X)}$ and $S^{-1}_I$ for the corresponding quantile function.
	
	Following \cite{CuiYangWu2013}, \cite{Assa2015}, \cite{ZhuangWengTanAssa2016}, and \cite{WengZhuang2017}, we assume that the insurance premium is given by
	\begin{equation}\label{eq:InsPrem}
		\pi_I = (1+\theta) \int_0^M g \big(S_I(x)\big) \, dx,
	\end{equation}
	in which $\theta \ge 0$ is the \emph{safety loading} and $g:[0,1] \to [0,1]$ is a strictly increasing \emph{distortion function} such that $g(0) = 0$ and $g(1) = 1$.\footnote{We assume that $g$ is \emph{strictly} increasing because we want to be able to compute its functional inverse.}  By Fubini's theorem, for any $\pi_I < \infty$,
	\begin{equation}\label{eq:InsPrem2}
		\begin{split}
			\frac{\pi_I}{1+ \theta} = \int_0^M \int_0^{S_I(x)} dg(p) \, dx = \int_0^{S_I(0)} \int_0^{S_I^{-1}(p)} dx \, dg(p) = \int_0^{S_X(0)} S_I^{-1}(p) \, dg(p).
		\end{split}
	\end{equation}
	Furthermore, for any $R \in \mathcal{A}$, $R(X)$ and $I(X) = X - R(X)$ are comonotonic because they are both non-decreasing functions of $X$. Thus, by the comonotonic additivity property of the distorted premium principle,
	\begin{equation}\label{eq:ComonotonicAdd}
		\pi_I = \pi_X - (1+\theta) \int_0^{S_X(0)} S_R^{-1}(p) \, dg(p),
	\end{equation}
	in which $\pi_X$, the price of insuring the total loss, is given by
	\begin{equation}\label{eq:piX}
		\pi_X = (1+\theta) \int_0^{S_X(0)} S_X^{-1}(p) \, dg(p),
	\end{equation}
	which we assume is finite.\footnote{\cite{WengZhuang2017} err in assuming $\pi_I + \pi_R = \pi_X$ because, in their paper, $I$ and $R$ are not necessarily comonotonic.  Only $R$ is required to be non-decreasing with respect to $X$, so $I = X - R$ might decrease as $X$ increases.  In fact, for the optimal insurance, namely, truncated deductible insurance, $I(x) = (x - d)_+ \, {\bf 1}_{\{x \le m\}}$, decreases to $0$ when $x > m$.  Also, if $0 < m < M$ in the expression for $I$, then $\pi_I + \pi_R \ne \pi_X$.}

	\begin{remark}
		Note that if $w \ge \pi_X$, in which $w > 0$ is the individual's initial capital, then the policyholder can avoid ruin with probability one by choosing the indemnity $I(x) = x$.  More generally, if
		\begin{equation}\label{eq:ws0}
			w \ge d + (1 + \theta) \int_d^M g \big(S_X(x) \big) \, dx,
		\end{equation}
		for some $d \ge 0$, then the individual can avoid ruin with probability one by purchasing deductible insurance with $I(x) = (x - d)_+$.  The least amount of money required for this the so-called {\rm safe level}, denoted by $w_s$ $($subscript {\rm s} for {\bf s}afe$)$, and is obtained by minimizing the right side of inequality \eqref{eq:ws0} with respect to $d \ge 0$.  It is straightforward to show that $w_s$ equals
		\begin{equation}\label{eq:ws}
			w_s = d_s + (1 + \theta) \int_{d_s}^M g \big(S_X(x) \big) \, dx,
		\end{equation}
		in which $d_s$ is given by
		\begin{equation}\label{eq:ds}
			d_s =
			\begin{cases}
				0, &\quad \text{if  } \theta \le \theta_s, \\
				S^{-1}_X \left( g^{-1} \left( \frac{1}{1 + \theta} \right) \right), &\quad \text{otherwise},
			\end{cases}
		\end{equation}
		and $\theta_s$ is given by
		\begin{equation}\label{eq:thetas}
			\theta_s = \frac{1}{g \left( S_X(0) \right)} - 1. 
		\end{equation}
		Henceforth, assume that $w < w_s$.   \qed
	\end{remark}

	We now formulate the optimal insurance problem from the perspective of the policyholder, who seeks to minimize her probability of ruin. For a given retention $R$, the net capital of the policyholder after paying the insurance premium and experiencing the loss retained loss is $w - \pi_I - R(X)$. Thus, the policyholder's probability of ruin equals
	\begin{equation}\label{eq:RuinProb}
		\PP \big(w -  \pi_I - R(X) < 0 \big) = \PP \big(R(X) > w -\pi_I\big) = S_R(w - \pi_I).
	\end{equation}
	It then follows from \eqref{eq:ComonotonicAdd} that the policyholder faces the following functional optimization problem,
	\begin{equation}\label{eq:OptIns}
		\inf_{R \in \mathcal{A}} S_R \Big(w - \pi_X + (1+\theta) \int_0^{S_X(0)} S_R^{-1}(p) \, dg(p) \Big),
	\end{equation}
	for $w < w_s$.
	
	\begin{remark}
		There is a duality between the ruin-minimization problem we consider and minimizing value-at-risk, or $VaR_\alpha$, given by
		\begin{equation}
			VaR_\alpha (R(X) + \pi_I) = \inf \{w \ge 0: \PP(R(X) + \pi_I > w) \le \alpha \}.
		\end{equation}
		Indeed, if the minimum probability of ruin in \eqref{eq:OptIns} equals $\alpha^*$ with minimizer $R^*$, then $w$ is the minimum value of $VaR_{\alpha^*}$ with minimizer $R^*$, and vice versa.  Thus, we expect DIML, or limited deductible insurance, to minimize the probability of ruin for $R \in \mathcal{A}$, as in Theorem $3.2$ of {\rm\cite{ChiTan2011}}.  Although our distortion premium principle is more general than their expected-value premium principle, our result should coincide with theirs qualitatively.   \qed
	\end{remark}

	\section{Reformulating \eqref{eq:OptIns} via the retention quantile function}\label{sec:reform}
	
	In problem \eqref{eq:OptIns}, the retention $R$ appears indirectly through the survival function $S_R$ and its right-continuous inverse $S^{-1}_R$.  It will be useful to reformulate \eqref{eq:OptIns} with $S_R^{-1}$ as the control variable. To this end, we introduce the set of admissible retention quantile functions. 
	
	\begin{definition}\label{def:quantilFn}
		A function $S_R^{-1}: [0, 1] \to [0, M]$ is called an {\rm admissible retention quantile function} if
		\begin{enumerate}
			\item[$(i)$] $S_R^{-1}$ is non-increasing and right-continuous;
			
			\item[$(ii)$] $S_R^{-1} \le S_X^{-1};$ and,
			
			\item[$(iii)$] $S_I^{-1} = S_X^{-1} - S_R^{-1}$ is also non-increasing and right-continuous.
		\end{enumerate}
		The set of all admissible retention quantile functions is denoted by $\mathcal{S}^{-1}$.  \qed
	\end{definition}

	\begin{remark}
		Any admissible retention quantile function $S_R^{-1} \in \mathcal{S}^{-1}$ must be {\rm continuous} on $[0, S_X(0)]$, not just {\rm right-continuous}. Continuity follows from $(iii)$ because any $($downward$)$ jump in the closed interval would make the indemnity quantile function $S_I^{-1} = S_X^{-1} - S_R^{-1}$ have an upward jump. This continuity is a reason to work with quantile functions instead of survival functions $($which potentially have jumps$)$.   \qed
	\end{remark}

	\begin{remark}\label{rem:IRinv_equiv}
		Note that $S_R^{-1}$ belongs to $\mathcal{S}^{-1}$ if and only if $S_I^{-1}$ belongs to $\mathcal{S}^{-1}$.   \qed
	\end{remark}

	Next, we establish a one-to-one correspondence between admissible retention functions in $\mathcal{A}$ and admissible retention quantile functions in $\mathcal{S}^{-1}$.
	
	\begin{lemma}\label{lem:R_to_Sr}
		For any $R \in \mathcal{A}$, $S_R^{-1}$ given by \eqref{eq:Sy_inv} belongs to $\mathcal{S}^{-1}$. Conversely, for any $S_R^{-1} \in \mathcal{S}^{-1}$, define $R:[0, M] \to [0, M]$ by
		\begin{equation}\label{eq:R_given_Sr}
			R(x) := S^{-1}_R \big( S_X(x) \big), 
		\end{equation}
		Then, $R \in\mathcal{A}$, and $S^{-1}_R$ is the inverse survival function of $R$.
	\end{lemma}
	
	\begin{proof}
		To show the first statement, we check that $S^{-1}_R$ satisfies conditions $(i)$-$(iii)$ of Definition \ref{def:quantilFn} for a given $R \in \mathcal{A}$. $S^{-1}_R$ is non-increasing and right-continuous by its definition in \eqref{eq:Sy_inv}. Also, $S_R(x) \ge S_X(x)$, for all $x \in [0,M]$, since $R(X) \le X$. Thus, by \eqref{eq:Sy_inv}, $S^{-1}_R \le S^{-1}_X$ on $[0, 1]$. It only remains to check condition $(iii)$. From condition $(iii)$ in Definition \ref{def:IndemnityFn}, we know that $I(x) := x - R(x)$ is non-decreasing with respect to $x$; thus, $R(X)$ and $I(X)$ are comonotonic, which implies that $S^{-1}_I + S^{-1}_R = S^{-1}_X$.  It follows that, if we define $S^{-1}_I$ by $S^{-1}_X - S^{-1}_R$, then $S^{-1}_I$ equals the inverse survival function of $I = X - R$ and is, hence, non-decreasing and right-continuous.
		
		To show the second statement, we check that $R$ satisfies conditions $(i)$-$(iii)$ of Definition \ref{def:IndemnityFn} for a given $S_R^{-1} \in \mathcal{S}^{-1}$. Note that $R$ defined by \eqref{eq:R_given_Sr} is non-decreasing because both $S^{-1}_R$ and $S_X$ are non-increasing.  Next, because $S^{-1}_R \le S^{-1}_X$ and because $S^{-1}_X$ maps $[0, S_X(0)]$ continuously onto $[0, M]$, it follows that $0 \le R(x) \le x$, for all $0 \le x \le M$.  Also, $R(X) = S^{-1}_R(S_X(X)) \sim S^{-1}_R(U)$, because $S_X(X) \sim U$, in which $U$ is uniformly distributed on $[0, 1]$.  Thus, the inverse survival function of $R$ is given by
		\begin{equation}
			\inf \left\{ x \ge 0: \PP \left(S^{-1}_R(U) > x \right) \le p \right\} = S^{-1}_R(p).
		\end{equation}

		It only remains to show that $R$ satisfies condition (iii) of Definition \ref{def:IndemnityFn}.  To that end, define $S_I^{-1}(p) := S_X^{-1}(p) - S_R^{-1}(p)$, $0 \le p \le 1$, and observe that $S_I^{-1} \in \mathcal{S}^{-1}$ by Remark \ref{rem:IRinv_equiv}.  By the above argument, $S^{-1}_I$ is the inverse survival function of $I(X)$, in which $I$ is defined by $I(x) := S^{-1}_I \big( S_X(x) \big)$.
		
		Finally, to show condition (iii) of Definition \ref{def:IndemnityFn}, it suffices to show that $I(X)+R(X)$ has the same distribution as $X$. The latter statement can be easily shown, since the comonotonicity of $I(X)$ and $R(X)$ implies that
		\[
		S^{-1}_{I(X)+R(X)}(p) = S^{-1}_I(p) + S^{-1}_R(p) = S^{-1}_X(p), \quad 0 \le p \le 1. \qedhere
		\]
	\end{proof}

	Let us illustrate the relationship between retentions and retention quantile functions by the following example. 
	
	\begin{example}\label{ex:DIML}
		Consider the retention
		\begin{equation}\label{eq:DIML_R}
			\tilde{R}(x) = \tilde{R}(x; d, m) =
			\begin{cases}
				x, &\quad 0 \le x \le d,\\
				d, &\quad d < x \le m,\\
				x - (m - d), &\quad x > m, 			
			\end{cases}	
		\end{equation}
		in which $0\le d \le m\le M$. Note that the corresponding indemnity, $\tilde{I}(x)=x-\tilde{R}(x)$, is \emph{deductible insurance with maximum limit (DIML)}, or {\rm limited deductible insurance}, given by
		\begin{equation}\label{eq:DIML_I}
			\tilde{I}(x) = \tilde{I}(x; d, m) = 
			\begin{cases}
				0, &\quad 0 \le x \le d,\\
				x - d, &\quad d < x \le m,\\
				m - d, &\quad x > m, 			
			\end{cases}
		\end{equation}
		in which $d$ is the {\rm deductible}, and $m$ is the {\rm maximum limit}.
		
		Let $p_i = S_X(i)$, for $i \in \{m, d\}$. Then, the quantile function of $\tilde{R}(X)$ is
		\begin{equation}\label{eq:DIML_Srinv}
			\tilde{S}_R^{-1}(p) = \tilde{S}_R^{-1}(p; p_m, p_d) = 
			\begin{cases}
				S^{-1}_X(p) - (m - d), &\quad 0 \le p < p_m,\\
				d, &\quad p_m \le p < p_d,\\
				S^{-1}_X(p), &\quad p_d \le p \le S_X(0),
			\end{cases}
		\end{equation} 
		which clearly belongs to $\mathcal{S}^{-1}$.
		
		Conversely, take arbitrary constants $0 \le p_m \le p_d\le S_X(0)$, set $i = S_X^{-1}(p_i)$, for $i \in \{m,d\}$, and define $\tilde{S}_R^{-1}\in\mathcal{S}^{-1}$ via \eqref{eq:DIML_Srinv}. A straightforward calculation shows that the retention given by \eqref{eq:R_given_Sr}, namely,
		\[
		\tilde{R}(x) := \tilde{S}^{-1}_R \big( S_X(x) \big), \quad x \in[0, M],
		\]
		coincides with \eqref{eq:DIML_R}. Thus, $\tilde{R}\in\mathcal{A}$ and $\tilde{S}^{-1}_R$ is the inverse survival function of $\tilde{R}(X)$.\qed
	\end{example}
	
	A direct consequence of Lemma \ref{lem:R_to_Sr} is that Problem \eqref{eq:OptIns} is equivalent to the following problem
	\begin{equation}\label{eq:OptIns2}
		\inf_{S_R^{-1}\in\mathcal{S}^{-1}} S_R\Big(w - \pi_X + (1+\theta) \int_0^{S_X(0)} S_R^{-1}(p) \, dg(p) \Big),
	\end{equation}
	in which, for each $S_R^{-1}\in\mathcal{S}^{-1}$, $S_R$ is its right-continuous inverse given by $S_R(x) = \inf\{p:S_R^{-1}(p)\le x\}$, $x\in[0,M]$. We state this result as the following proposition, whose proof directly follows from Lemma \ref{lem:R_to_Sr}, so we omit it.

	\begin{proposition}
		Problems \eqref{eq:OptIns} and \eqref{eq:OptIns2} are equivalent. Specifically, if $S^{-1,*}_R \in \mathcal{S}^{-1}$ optimizes problem \eqref{eq:OptIns2}, then an optimal solution of \eqref{eq:OptIns} is given by
		\begin{equation}\label{eq:r*_given_Srinv*}
			R^*(x) := S^{-1,*}_R \big( S_X(x) \big), \quad x\in[0,M].
		\end{equation}
		Conversely, if $R^*\in\mathcal{A}$ solves problem \eqref{eq:OptIns}, then $S^{-1,*}_R$, the quantile function of $R^*(X)$, belongs to $\mathcal{S}^{-1}$ and solves problem \eqref{eq:OptIns2}.  \qed
	\end{proposition}
	
	
	The following theorem shows that the optimal insurance policy is a DIML contract, as introduced in Example \ref{ex:DIML}.
	
	\begin{theorem}\label{prop:Optimality_of_DIML}
		If \eqref{eq:OptIns2} has an optimal solution, then there exists an optimal solution of \eqref{eq:OptIns} that is a DIML retention function given by \eqref{eq:DIML_R}, for some $0 \le d \le m \le M$.
	\end{theorem}
	
	\begin{proof}
		Let $S_{R}^{-1,*}$ be an optimal solution of \eqref{eq:OptIns2}, define $d^*$ by
		\begin{equation}\label{eq:aStar}
			d^*= w - \pi_X + (1 + \theta) \int_0^{S_X(0)} S_{R}^{-1, *}(p) \, dg(p),
		\end{equation}
		and set $p_m^* = S^*_{R}(d^*)$, in which $S^*_{R}$ is the right-continuous inverse of $S^{-1,*}_R$.  Note that $p^*_m$ is the minimum probability of ruin.
		
		Consider the following auxiliary problem:
		\begin{equation}\label{eq:OptIns_aux1}
			\sup_{S_R^{-1}\in\mathcal{S}^{-1}} \left\{ \int_0^{S_X(0)} S_R^{-1}(p) \, dg(p) : S_R^{-1}(p^*_m) = d^* \right\}.
		\end{equation}
		We first show that any optimal solution of \eqref{eq:OptIns_aux1} is also an optimal solution of \eqref{eq:OptIns2}. Let $S_{R}^{-1,A}\in\mathcal{S}^{-1}$ be an optimal solution of \eqref{eq:OptIns_aux1}, and denote its right-continuous inverse by $S^A_R$.  Then, because $S^*_R(d^*) = p^*_m$, it follows that
		\begin{equation}\label{eq:premineq}
			\int_0^{S_X(0)} S_R^{-1, A}(p) \, dg(p) \ge  \int_0^{S_X(0)} S_R^{-1, *}(p) \, dg(p).
		\end{equation}
		Define
		\begin{equation}\label{eq:defdA}
			d^A = w - \pi_X + (1 + \theta)  \int_0^{S_X(0)} S_R^{-1, A}(p) \, dg(p).
		\end{equation}
		Inequality \eqref{eq:premineq} implies $d^A \ge d^*$; thus, $S^A_R(d^A) \le S^A_R(d^*) = p^*_m$.  Because $S^{-1,*}_R$ is an optimal solution of \eqref{eq:OptIns2}, it follows that $S^A_R(d^A) = p^*_m$, so $S^{-1,A}_R$ is also an optimal solution for \eqref{eq:OptIns2}.

		To finish the proof, it only remains to show that the solution of \eqref{eq:OptIns_aux1} is a DIML retention quantile function, as given in \eqref{eq:DIML_Srinv}. Figure \ref{fig:Feasibility} illustrates the feasibility set of problem \eqref{eq:OptIns_aux1}.  In particular, \eqref{eq:OptIns_aux1} looks for a non-increasing and continuous function $S^{-1}_R:[0, S_X(0)] \to [0, M]$ whose graph falls into the shaded region of Figure \ref{fig:Feasibility} such that the integral $\int_0^{S_X(0)} S^{-1}_R(p) \, dg(p)$ is maximized.  Recall that an admissible retention $S^{-1}_R$ is such that $S^{-1}_X - S^{-1}_R$ is non-decreasing.  It is evident that the optimal choice for $S^{-1}_R$ is the upper boundary of the shaded region, which coincides with a DIML retention quantile function \eqref{eq:DIML_Srinv}.
	\end{proof}
	
	\begin{remark}
		From the proof of Theorem {\rm \ref{prop:Optimality_of_DIML}}, we learn that, for the optimal DIML retention, $\tilde{R}^*$, the probability $p_m^* = S_{\tilde R^*}(d^*) = S_X(m^*)$ is the minimum probability of ruin.   \qed
	\end{remark}

	\begin{figure}[t]
		\centerline{
			\adjustbox{trim={0.05\width} {0.0\height} {0.04\width} {0.06\height},clip}{\includegraphics[scale=0.35,page=1]{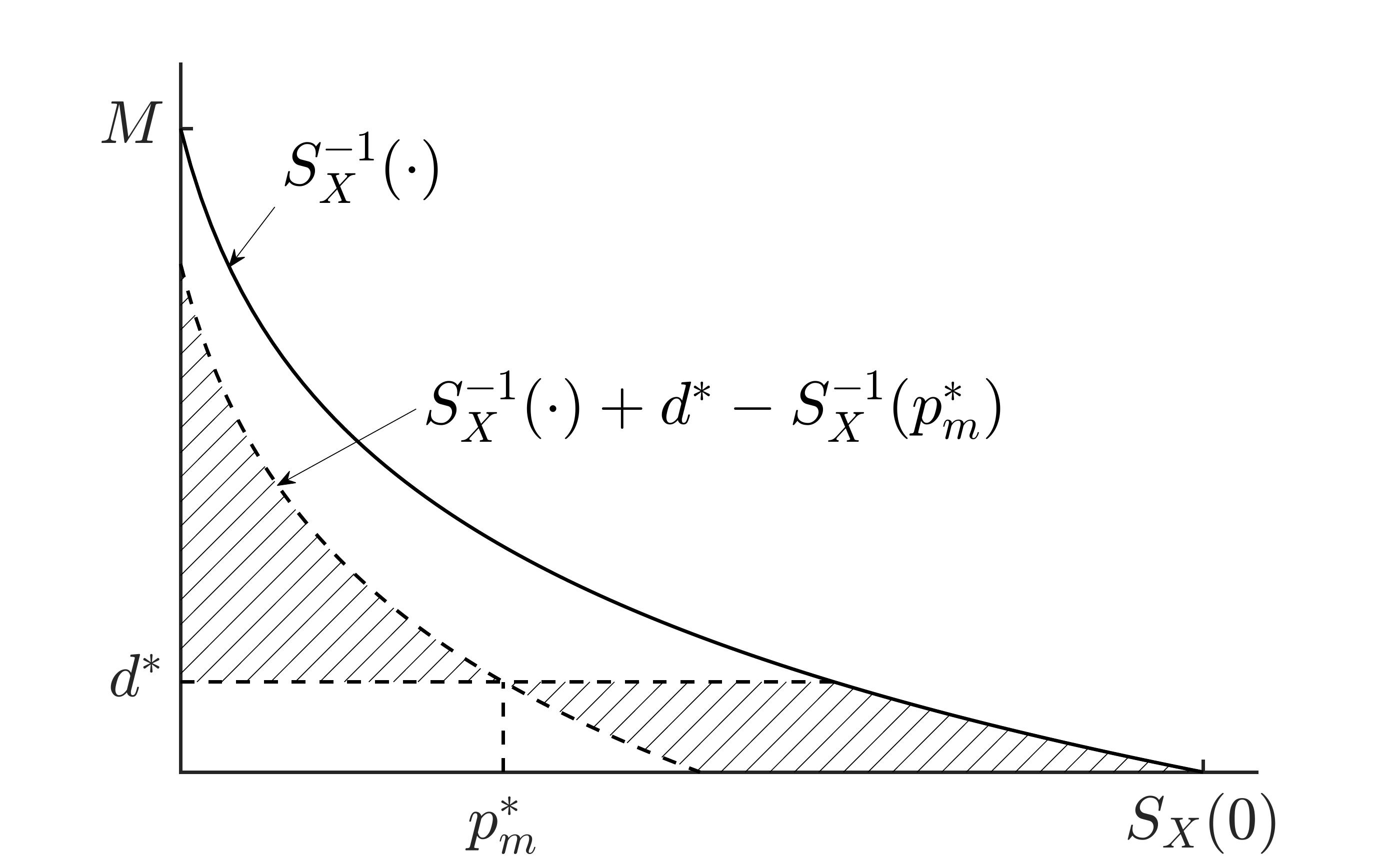}}
		}
		\caption{The feasible region for problem \eqref{eq:OptIns_aux1} includes any non-increasing and continuous function $S_R^{-1}$ whose graph falls into the shaded area such that $S_I^{-1}=S_X^{-1}-S_R^{-1}$ is also decreasing.}
		\label{fig:Feasibility}
	\end{figure}

	\section{Optimal DIML contract}\label{sec:opt}
	
	Theorem \ref{prop:Optimality_of_DIML} implies that if \eqref{eq:OptIns2} has an optimal solution, then it is sufficient to look within the class of DIML retention quantiles.  Thus, our next goal is to find the optimal DIML contract.  By restricting indemnity contracts to those given by \eqref{eq:DIML_I}, problem \eqref{eq:OptIns} becomes a two-dimensional problem:
	\begin{equation}\label{eq:OptIns_DIML}
		\inf_{\substack{0 \le p_m \le p_d \le S_X(0) \\
				p_d \ge S_X(w)}} \tilde{S}_R \Big(w - \pi_X + (1 + \theta) \int_0^{S_X(0)} \tilde{S}_R^{-1}(p; p_m, p_d) \, dg(p); \, p_m, p_d \Big),
	\end{equation}
	in which $\tilde{S}_R^{-1}$ is given by \eqref{eq:DIML_Srinv}, and $\tilde{S}_R$ is its right-continuous inverse.  Recall that $p_i = S_X(i)$ for $i \in \{ m, d \}$.  Note that we can restrict $d$ to be such that $d \le w$, or equivalently $p_d \ge S_X(w)$, because buying no insurance is always better than buying a DIML contract with $d > w$.

	
	The following theorem provides the optimal (DIML) insurance policy.
	
	\begin{theorem}\label{thm:OptimalDIML}
		Assume that initial wealth $w$ is less than the safe level $w_s$, as given by \eqref{eq:ws}. An optimal insurance policy to minimize the probability of ruin is given by a DIML contract, and the optimal deductible $d^*$ and maximum limit $m^*$ for the optimal DIML insurance policy are as follows . First, recall that $d_s$ is given by
		\begin{equation}\label{eq:dstar}
			d_s =
			\begin{cases}
				0, &\text{if  } \theta \le \theta_s = \frac{1}{g \left( S_X(0) \right)} - 1, \\
				S^{-1}_X \left( g^{-1} \left( \frac{1}{1 + \theta} \right) \right), &\text{otherwise}.
			\end{cases}
		\end{equation}
		\begin{enumerate}
			\item[(i)] If $0 \le \theta \le \theta_s$, then $d^* = d_s = 0$, which is independent of initial wealth, and $m^* < M$ uniquely solves
			\begin{equation}\label{eq:mStar_d0}
				w = (1 + \theta) \int_0^m g \big( S_X(t) \big) \, dt.
			\end{equation}
			
			\item[(ii)] If $\theta > \theta_s$ and $w \le d_s$, then it is optimal not to buy insurance, or equivalently, $d^* = m^* = w$.
			
			\item[(iii)] If $\theta > \theta_s$ and $w > d_s$, then $d^* = d_s > 0$, which is independent of initial wealth, and $m^* < M$ uniquely solves
			\begin{equation}\label{eq:mStar_dpos}
				w - d_s = (1 + \theta) \int_{d_s}^m g \big( S_X(t) \big) \, dt.
			\end{equation}
		\end{enumerate}
	\end{theorem}
	\begin{proof}
		See Appendix \ref{app:OptimalDIML}.
	\end{proof}
	
	Theorem \ref{thm:OptimalDIML} can be interpreted as follows. If insurance is cheap enough, namely, $\theta \le \theta_s$, then the policyholder, regardless of the amount of her wealth, should spend all of her initial wealth to buy a DIML contract with a zero deductible and a maximum limit less than $M$.  Thus, she ruins if the loss exceeds the maximum limit.
	
	However, if the insurance is expensive (that is, $\theta > \theta_s$), then the optimal policy changes according to how wealthy the policyholder is.  If her wealth is less than $d_s > 0$, then she is better off buying no insurance.  If wealth is between $d_s$ and $w_s$, it is optimal to buy a DIML policy with non-zero deductible and a maximum limit less than $M$. That policy is such that her initial wealth exactly covers the insurance premium with remainder equal to the deductible.  Thus, she ruins if the loss exceeds the maximum limit.
	
	We also have a corollary that follows readily from Theorem \ref{thm:OptimalDIML}.
	
	\begin{corollary}\label{cor:binds}
		For the optimal DIML contract given in Theorem \ref{thm:OptimalDIML}, the deductible plus the premium equals initial wealth, that is, $d^* + (1 + \theta) \int_{d^*}^{m^*} g \big( S_X(t) \big) \, dt = w$.  In other words, the constraint in \eqref{eq:Aux1} binds.  Moreover, the minimum probability of ruin equals $S_X(m^*)$.   \qed
	\end{corollary}
	
	Recall that, given any optimal solution of problem \eqref{eq:OptIns2}, the proof of Theorem \ref{prop:Optimality_of_DIML} computes a quantity $d^*$ in \eqref{eq:aStar}, which, in turn, becomes the deductible for a DIML contract with the same minimum probability of ruin.  Suppose we begin with the optimal DIML contract, as given in Theorem \ref{thm:OptimalDIML}, with a deductible $d^*$ and and maximum limit $m^*$, then compute the expression on the right side of \eqref{eq:aStar}.  Because the expression on the right side of \eqref{eq:aStar} equals initial wealth minus the premium, it follows from Corollary \ref{cor:binds} that we obtain the deductible $d^*$ of the original DIML policy.  Corollary \ref{cor:binds} also implies that the computation of $p^*_m$, the minimum probability of ruin, in the proof of \ref{prop:Optimality_of_DIML} leads to the same maximum limit $m^*$ as in the original DIML policy.

	\appendix
	
	\section{Proof of Theorem \ref{thm:OptimalDIML}}\label{app:OptimalDIML}
	
	We introduce the following notation. Define the auxiliary functions $\Psi$ and $\Phi$ by
	\begin{align}
		\Psi(x) &= (1 + \theta) \int_x^M g \big( S_X(t) \big) \,dt, \label{eq:Psi}
		\intertext{and}
		\Phi(x) &= x + \Psi(x), \label{eq:Phi}
	\end{align}
	for $x \in [0, M]$.  Note that $\Psi$ is strictly decreasing; denote its inverse by $\Psi^{-1}$.
	
	\begin{remark}
		It is easy to show that $\Psi(d) - \Psi(m)$ is the premium corresponding to a DIML insurance with deductible $d$ and maximum level $m$; see \eqref{eq:Aux2} below.  This result is well known to those who work with the distortion premium principle, but we include it for completeness.  Also, note that $w_s = \Phi(d_s)$, in which $w_s$ and $d_s$ are given in \eqref{eq:ws} and \eqref{eq:ds}, respectively.
	\end{remark}

	For ease of reference in this proof, we present the survival function $\tilde{S}_R$ for the retention function of DIML insurance.  $\tilde{S}_R$, the right-continuous inverse of $\tilde{S}^{-1}_R$ in \eqref{eq:DIML_Srinv}, is given by
	\begin{equation}\label{eq:DIML_Sr}
		\tilde{S}_R(x) = 
		\begin{cases}
			S_X(x), &\quad 0 \le x < d,\\
			S_X(x + m - d), &\quad d \le x \le M.
		\end{cases}
	\end{equation} 
	
	Without loss of generality, we can assume that the optimal insurance premium $\pi^*_I$, for the optimal indemnity and the optimal deductible $d^*$ satisfy
	\begin{equation}\label{eq:Aux1}
		w - \pi^*_{I} \ge d^*.
	\end{equation}
	Indeed, if $w -\pi^*_I < d^*$, then \eqref{eq:DIML_Sr} yields that the minimum probability of ruin equals
	\[
	S_{\tilde{R}^*(X)}(w - \pi^*_{I}) = \tilde{S}^*_R(w - \pi^*_I) = S_X(w - \pi^*_{I}) \ge S_X(w),
	\]
	in which $\tilde{R}^*(X)$ is the optimal DIML retention, or equivalently, $\tilde{S}^*_R = S_{\tilde{R}^*(X)}$ is the optimal survival function of the retention.  But, this inequality implies that buying no insurance is also optimal, in which case, we can set $d^* = m^* = w$ to recover \eqref{eq:Aux1}. Thus, imposing the inequality
	\begin{equation}\label{eq:nonBinding}
		w - \pi_I \ge d
	\end{equation}
	will not change the optimal value of problem \eqref{eq:OptIns_DIML}. Under this additional constraint, by \eqref{eq:DIML_Sr}, the probability of ruin simplifies to
	\[
	\tilde{S}_R(w - \pi_I) = S_X(w - \pi_I + m - d).
	\]
	Thus, problem \eqref{eq:OptIns_DIML} is equivalent to
	\begin{equation}\label{eq:OptIns_DIML_v2}
		\sup_{0 \le d \le m \le M} \big\{ m - d - \pi_I : d + \pi_I \le w \big\}.
	\end{equation}

	We can simplify this optimization problem further. By \eqref{eq:ComonotonicAdd}, \eqref{eq:piX}, and \eqref{eq:DIML_Srinv}, we obtain
	\[
	\begin{split}
		\frac{\pi_I}{1 + \theta} = \frac{\pi_X - \pi_R}{1 + \theta} &= \int_0^{S_X(0)} S_X^{-1}(p) \, dg(p) - \int_0^{S_X(0)} \tilde{S}_R^{-1}(p;p_m,p_d) \, dg(p)\\
		&= \int_0^{p_m} \big(m - S^{-1}_X(p) \big) \, dg(p) - \int_0^{p_d} \big(d - S_X^{-1}(p) \big) \, dg(p) \\
		&= \int_{p_d}^{p_m} g(p) \, d S_X^{-1}(p) = \int_d^m g \big( S_X(t) \big) \, dt.
	\end{split}
	\]
	Therefore,
	\begin{equation}\label{eq:Aux2}
		\pi_I = \Psi(d) - \Psi(m),
	\end{equation}
	in which $\Psi$ is given by \eqref{eq:Psi}, and \eqref{eq:OptIns_DIML_v2} becomes
	\begin{equation}\label{eq:PhiPsi_Optim}
		\sup_{d, m} \left\{ \Phi(m) - \Phi(d) : 0 \le d \le m \le \Psi^{-1}\big( \Phi(d) - w \big) \right\},
	\end{equation}
	in which $\Phi$ is given by \eqref{eq:Phi}.  Note that $\Phi(d) - w > 0$ for all $d \in [0, M]$ because $w < w_s = \min_{d \in [0, M]} \Phi(d)$.

	We need the following result regarding monotonicity of $\Psi$ and $\Phi$.  Its proof is immediate; thus, we omit it. 
	\begin{lemma}\label{lem:PsiPhiMonot}
		For $x \in (0, M)$,
		\begin{equation}\label{eq:PsiPhi_prime}
			\Psi^\prime(x) = - (1 + \theta) g\big(S_X(x) \big), \quad \text{and} \quad
			\Phi^\prime(x) = 1 -  (1 + \theta) g\big(S_X(x) \big).
		\end{equation}
		In particular, $\Psi$ is strictly decreasing on $[0, M]$, and $\Phi$ is strictly decreasing $($respectively, strictly increasing$)$ on $[0, d_s)$ $($respectively, $(d_s, M])$, in which $d_s$ is given by \eqref{eq:ds}.
	\end{lemma}
	
	For a fixed $d \in [0, M]$, Lemma \ref{lem:PsiPhiMonot} implies that
	\begin{equation}\label{eq:Pm_Optim}
		\begin{split}
			&\sup_{m} \left\{ \Phi(m) : d \le m \le \Psi^{-1} \big( \Phi(d) - w \big) \right\} \\
			&= \max \left\{ \Phi(d), \, \Phi \left( \Psi^{-1} \left( \Phi(d) - w \right) \right) \right\}.
		\end{split}
	\end{equation}
	In other words, the best maximum limit $m$, given a fixed deductible $d \in [0, M]$, is either $m = d$ or
	\begin{equation}\label{eq:m_Given_d}
		m = \Psi^{-1} \big( \Phi(d) - w \big),
	\end{equation}
	in which the former choice means buying no insurance.
	
	By substituting \eqref{eq:Pm_Optim} into \eqref{eq:PhiPsi_Optim}, and by noting that 
	$$
	\Phi \Big( \Psi^{-1}\big(\Phi(w) - w \big) \Big) - \Phi(w) = 0,
	$$
	we obtain the following single-variable optimization problem:
	\begin{equation}\label{eq:Pm_Optim2}
		\sup_{0 \le d \le w} \left\{ \Phi \Big(\Psi^{-1}\big(\Phi(d) - w \big)\Big) - \Phi(d) \right\}.
	\end{equation}
	Let $d^*$ be the optimal deductible for problem \eqref{eq:Pm_Optim2}. We have the following dichotomy:
	\begin{enumerate}
		\item[(a)] If $\Phi \left(\Psi^{-1}\big(\Phi(d^*) - w \big) \right) - \Phi(d^*) > 0$, then, by \eqref{eq:m_Given_d}, the maximum limit $m^*$ is given by
		\begin{equation}\label{eq:mStar_Given_dStar}
			m^* = \Psi^{-1} \big( \Phi(d^*) - w \big).
		\end{equation}
		
		\item[(b)] If $\Phi \left(\Psi^{-1}\big(\Phi(d^*) - w \big) \right) - \Phi(d^*) = 0$, then, we can choose $d^* = m^* = w$, and buying no insurance is the best choice for the policyholder.
	\end{enumerate}

	\noindent It only remains to solve \eqref{eq:Pm_Optim2}. From \eqref{eq:PsiPhi_prime}, we obtain
	\[
	\frac{d}{d d} \left[ \Phi \Big(\Psi^{-1}\big(\Phi(d) - w \big)\Big) - \Phi(d) \right] = - \frac{\Phi'(d)}{(1 + \theta) g \left( S_X \left( \Phi^{-1} \left( \Phi(d) - w \right) \right) \right)} \propto - \Phi'(d).
	\]
	Lemma \ref{lem:PsiPhiMonot}, then, yields that the objective function of \eqref{eq:Pm_Optim2} is strictly increasing on $[0, d_s)$ and strictly decreasing on $(d_s, w]$, with $d_s$ given by \eqref{eq:ds}.  It follows that $d_s$ is the optimal deductible if $d_s < w$.  Recall that we impose the condition that $d + \pi_I \le w$; in particular, $d < w$ if the individual buys insurance.
	
	We identify the following cases:
	\begin{enumerate}
		\item[(i)] $0 \le \theta \le \theta_s$:  By \eqref{eq:ds}, $d^* = d_s = 0$, and case (a) above applies. By \eqref{eq:mStar_Given_dStar},
		\[
		m^* = \Psi^{-1} \big( \Phi(0) - w \big) = \Psi^{-1} \big( \pi_X - w \big),
		\]
		or equivalently, $m^*$ uniquely solves \eqref{eq:mStar_d0}.  $m^* < M$ because $w < w_s = \Phi(0) = \pi_X$.
		
		\item[(ii)] $\theta > \theta_s$ and $w \le d_s$:  It follows that the objective function of \eqref{eq:Pm_Optim2} is strictly increasing on $[0, w)$. Thus, case (b) above holds, and it is optimal for the individual not to buy insurance.
		
		\item[(iii)] $\theta > \theta_s$ and $w > d_s$:  By \eqref{eq:ds}, $d^* = d_s > 0$, and by \eqref{eq:mStar_Given_dStar},
		\[
		m^* = \Psi^{-1} \big( \Phi(d_s) - w \big),
		\]
		or equivalently, $m^*$ uniquely solves \eqref{eq:mStar_dpos}.  $m^* < M$ because $w < w_s = \Phi(d_s)$, which implies that $w - d_s < (1 + \theta) \int_{d_s}^M g \big( S_X(t) \big) \, dt$.
		
		\qedhere
	\end{enumerate}

	

\end{document}